\newif\iflncs
\lncsfalse
\newif\ifnotes
\notestrue

\iflncs
\documentclass[orivec,runningheads,envcountsect,envcountsame,a4paper]{llncs}
\else
\documentclass[11pt]{article}
\fi

\iflncs
\usepackage{breakcites}
\else
\usepackage{amsthm}
\usepackage{fullpage}
\fi

\usepackage{subcaption}
\usepackage{algorithm}
\usepackage[noend]{algpseudocode}
\usepackage{amsfonts}
\usepackage{amsmath}
\usepackage{bbm}
\usepackage{complexity}
\usepackage{mdframed}
\usepackage{framed}
\usepackage{float}
\usepackage{array}
\usepackage{enumitem}
\usepackage[dvipsnames]{xcolor} 
\usepackage{graphicx}
\usepackage{tikz}
\usepackage{pgfplots}
\usepackage{multirow}
\usepackage{pdfpages}
\usepackage{booktabs}
\iflncs
\usepackage[colorlinks=true,linkcolor=blue,citecolor=blue]{hyperref}
\else
\usepackage[colorlinks=true,linkcolor=Maroon,citecolor=Maroon]{hyperref}
\fi
\usepackage{thm-restate}
\usepackage{orcidlink}                    
\usepackage[capitalize]{cleveref}  
\hypersetup{hypertexnames=false}

\newcommand{\remove}[1]{}

\ifnotes
\newcommand{\knote}[1]{[{\footnotesize \textsf{\color{orange}{\bf Kel Zin:} { {#1}}}}]}
\definecolor{ao}{rgb}{0.0, 0.5, 0.0}
\newcommand{\pnote}[1]{[{\footnotesize \textsf{\color{cyan}{\bf Prashant:} { {#1}}}}]}

\else
\newcommand{\knote}[1]{}
\definecolor{ao}{rgb}{0.0, 0.5, 0.0}
\newcommand{\pnote}[1]{}
\fi
\definecolor{mygreen}{RGB}{0,128,0} 
\newcommand{\resolved}[1]{{\color{mygreen}{[Resolved]}}}

\iflncs
\input{lncs_issues}

\let\manualqed\qed
\let\qedhere\relax
\spnewtheorem{claim}[theorem]{Claim}{\bfseries}{\itshape}
\else
\let\manualqed\relax
\newtheorem{theorem}{Theorem}[section]

\newtheorem{claim}[theorem]{Claim}
\newtheorem{corollary}[theorem]{Corollary}

\newtheorem{lemma}[theorem]{Lemma}
\newtheorem{proposition}[theorem]{Proposition}
\newtheorem{definition}[theorem]{Definition}
\theoremstyle{definition}

\fi

\iflncs
\renewcommand{\paragraph}[1]{\;\newline \noindent \textbf{#1}}
\fi

\crefname{definition}{Definition}{Definitions}
\crefname{sub-definition}{Definition}{Definitions}
\crefname{example}{Example}{Examples}
\crefname{exercise}{Exercise}{Exercises}
\crefname{property}{Property}{Properties}
\crefname{question}{Question}{Questions}
\crefname{solution}{Solution}{Solutions}
\crefname{theorem}{Theorem}{Theorems}
\crefname{informaltheorem}{Theorem}{Theorems}
\crefname{proposition}{Proposition}{Propositions}
\crefname{problem}{Problem}{Problems}
\crefname{lemma}{Lemma}{Lemmas}
\crefname{conjecture}{Conjecture}{Conjectures}
\crefname{corollary}{Corollary}{Corollaries}
\crefname{fact}{Fact}{Facts}
\crefname{claim}{Claim}{Claims}
\crefname{remark}{Remark}{Remarks}
\crefname{note}{Note}{Notes}
\crefname{figure}{Figure}{Figure}
\crefname{case}{Case}{Cases}
\crefname{proofsketch}{Proof Sketch}{Proof Sketches}
\crefname{pseudocodeline}{line}{lines}

\def\ddefloop#1{\ifx\ddefloop#1\else\ddef{#1}\expandafter\ddefloop\fi}

\def\ddef#1{\expandafter\def\csname bb#1\endcsname{\ensuremath{\mathbb{#1}}}}
\ddefloop ABCDEFGHIJKLMNOPQRSTUVWXYZ\ddefloop

\def\ddef#1{\expandafter\def\csname c#1\endcsname{\ensuremath{\mathcal{#1}}}}
\ddefloop ABCDEFGHIJKLMNOPQRSTUVWXYZ\ddefloop

\def\ddef#1{\expandafter\def\csname v#1\endcsname{\ensuremath{\overline{#1}}}}
\ddefloop ABCDEFGHIJKLMNOPQRSTUVWXYZabcdefghijklmnopqrstuvwxyz\ddefloop

\newcommand{\algfont}[1]{\mathsf{#1}}
\def\ddef#1{\expandafter\def\csname alg#1\endcsname{\ensuremath{\algfont{#1}}}}
\ddefloop ABCDEFGHIJKLMNOPQRSTUVWXYZ\ddefloop

\newcommand{\langfont}[1]{\mathnormal{#1}}
\def\ddef#1{\expandafter\def\csname lang#1\endcsname{\ensuremath{\langfont{#1}}}}
\ddefloop ABCDEFGHIJKLMNOPQRSTUVWXYZ\ddefloop

\def\ddef#1{\expandafter\def\csname v#1\endcsname{\ensuremath{\boldsymbol{\csname #1\endcsname}}}}
\ddefloop {alpha}{beta}{gamma}{delta}{epsilon}{varepsilon}{zeta}{eta}{theta}{vartheta}{iota}{kappa}{lambda}{mu}{nu}{xi}{pi}{varpi}{rho}{varrho}{sigma}{varsigma}{tau}{upsilon}{phi}{varphi}{chi}{psi}{omega}{Gamma}{Delta}{Theta}{Lambda}{Xi}{Pi}{Sigma}{varSigma}{Upsilon}{Phi}{Psi}{Omega}{ell}\ddefloop

\newcommand{\eps}{\varepsilon}

\newcommand{\Reals}{\mathbb{R}}
\newcommand{\Field}[1]{\mathbb{F}_{#1}}
\newcommand{\Integer}[1]{\mathbb{Z}_{#1}}

\DeclareMathOperator*{\Exp} {{\mathbf{E}}}
\DeclareMathOperator*{\Var} {{\mathbf{Var}}}

\newcommand{\prob}[2]       {\Pr_{#1}\left[ #2 \right]}
\let\exponential\exp
\renewcommand{\exp}[1]      {\Exp\left[ #1 \right]}
\newcommand{\expec}[2]      {\Exp_{#1}\left[ #2 \right]}

\newcommand{\round}[1]        {\left\lfloor#1\right\rceil}

\renewcommand{\poly}{\mathrm{poly}}
\renewcommand{\polylog}{\mathrm{polylog}}

\newcommand{\abs}[1]        {\left| #1\right|}

\newcommand{\floor}[1]      {\left\lfloor #1 \right\rfloor}
\newcommand{\ceil}[1]       {\left\lceil #1 \right\rceil}
\newcommand{\ip}[1]         {\langle #1 \rangle}

\usepackage{listofitems}
\newcommand\cycle[2][\,]{%
  \readlist\thecycle{#2}%
  (\foreachitem\i\in\thecycle{\ifnum\icnt=1\else#1\fi\i})%
}

\newcounter{pseudocodeline}
\crefname{pseudocodeline}{line}{lines}

\newcommand{\pseudoline}[1][]{%
  \refstepcounter{pseudocodeline}%
  \small\thepseudocodeline%
  \if\relax\detokenize{#1}\relax
  \else
    \label{#1}%
  \fi
  \>%
}

\newenvironment{pseudocode}{
  \setcounter{pseudocodeline}{0}
  \begin{tabbing}
    \qquad\=\qquad\=\qquad\=\qquad\=\qquad\=\qquad\=\qquad\= \kill
}{
  \end{tabbing}
}

\title{Sample-Preserving Search-to-Decision Reduction\\for Noisy Linear Equations over Large Moduli}

\iflncs
\author{}
\institute{}
\date{}
\else
\author{Andrej Bogdanov\footnote{\orcidlink{0000-0002-0338-6151} {\tt abogdano@uottawa.ca}. University of Ottawa.} 
\and Kel Zin Tan\footnote{\orcidlink{0009-0008-1099-0892} {\tt kelzin@u.nus.edu}. Department of Computer Science, National University of Singapore.} 
\and Prashant Nalini Vasudevan\footnote{\orcidlink{0000-0001-6880-795X} {\tt prashvas@nus.edu.sg}. Department of Computer Science, National University of Singapore.}}
\date{\today}
\fi

\begin{document}

\pagenumbering{roman}

\maketitle
\begin{abstract}

The Noisy Linear Equations problem involves finding the solution to a random linear system over a finite field given noisy evaluations. This is a generalisation of various learning problems widely used in cryptography, including Learning with Errors (LWE), Learning with Rounding (LWR), and Learning Parity with Noise (LPN).

We present a new sample-preserving search-to-decision reduction for Noisy Linear Equations that has complexity $\mathrm{poly}(n,m,\log q, q/\sigma)$, where $n$ is the dimension of the secret, $m$ is the number of samples, $q$ is the modulus, and $\sigma$ is the magnitude of the noise (defined suitably for the kind of noise involved). In particular, this implies search-to-decision reductions for LWE and LWR that run in time $\mathrm{poly}(n)$ even if the modulus $q$ is exponentially large, as long as the noise-to-modulus ratio $\sigma/q$ is non-negligible. Our reduction works for prime moduli $q$. If $q$ is composite, it instead produces a list that contains a small multiple of the solution.

Prior to this work, known search-to-decision reductions for LWE and LWR either did not preserve the parameters of the problem (i.e., $m$, $n$, $\sigma$, and $q$), or had complexity that scaled polynomially with either $\sigma$ or the largest prime factor of $q$.

\end{abstract}

\thispagestyle{plain}

\iflncs
\else
\newpage
\setcounter{tocdepth}{2}
\tableofcontents
\newpage
\fi

\pagenumbering{arabic}
\section{Introduction} 

The \emph{Noisy Linear Equations} problem asks to recover the solution to a random system of linear equations over $\Integer{q}$ given noisy outputs. Let $\cR : \Integer{q}^m \to \Integer{p}^m$ be a ``noise'' function. (This function could be randomised, and $p$ may or may not be the same as $q$.) Given a random matrix $A \in \Integer{q}^{m \times n}$ and the noisy evaluation $\cR(Ax)$ for a random $x\in\Integer{q}^n$, the goal for the \emph{Search} problem is to recover the solution $x$. For the \emph{Decision} problem, similarly given $A$, the goal is to distinguish the noisy evaluations $\cR(Ax)$ from a random vector in $\Integer{p}^m$.

This problem is clearly easy for some choices of the noise function $\cR$, such as the identity function. It is believed to be hard for certain choices of the noise function, such as the following:
\begin{itemize}
\item $\cR(v) = v + e  \text{ mod } q$, where each entry in $e$ is draws from the discrete Gaussian distribution. This is known as the \emph{Learning With Errors (LWE)} problem \cite{DBLP:journals/jacm/Regev09}. The hardness of LWE can be shown based on the worst-case hardness of lattice problems \cite{DBLP:journals/jacm/Regev09, DBLP:conf/stoc/Peikert09, DBLP:conf/stoc/BrakerskiLPRS13}
\item $\cR(v) = \round{ p\cdot v / q} \text{ mod } p$, where $p < q$ and $\round{\cdot}$ represents rounding to the nearest integer. This is called the \emph{Learning With Rounding (LWR)} problem \cite{DBLP:conf/eurocrypt/BanerjeePR12}, whose hardness follows from reductions from LWE \cite{DBLP:conf/eurocrypt/BanerjeePR12, DBLP:conf/crypto/AlwenKPW13, DBLP:conf/tcc/BogdanovGMRR16, DBLP:journals/iacr/Alperin-Sheriff16}.
\item $\cR(v)_i = v_i + a_i\cdot b_i$ for each $i\in[m]$, where $a_i$ is a random element of $\Integer{q}$ and $b_i$ is a biased random bit. This is the generalised \emph{Learning Parity with Noise (LPN)} problem \cite{DBLP:conf/crypto/BlumFKL93}. 
\end{itemize}
The above problems are of considerable importance to cryptography, and form the basis of numerous constructions of advanced cryptographic primitives~\cite{DBLP:journals/jacm/Regev09, DBLP:conf/focs/BrakerskiV11, DBLP:conf/eurocrypt/BanerjeePR12, BLMR13, BP14}.

\paragraph{Search-to-Decision Reductions.} In the regimes of interest to public-key cryptography, the planted solution $x$ is typically unique.  Search is thus at least as hard as decision.  A search-to-decision reduction establishes their equivalence:  It efficiently converts a distinguisher (decision algorithm) into a solution finder (search algorithm).  Cryptographic applications such as public-key encryption \cite{DBLP:journals/jacm/Regev09}, fully homomorphic encryption \cite{DBLP:conf/focs/BrakerskiV11}, and structured psedurorandom functions \cite{DBLP:conf/eurocrypt/BanerjeePR12, BLMR13, BP14} assume decision hardness.  A search-to-decision reduction increases our confidence in their security. In order for this equivalence to be as tight as possible, it is desirable for such reductions to be computationally efficient, preserve the number of samples and variables in the problem instance (i.e., the parameters $m$ and $n$), and preserve the noise function $\cR$.

\medskip
Recall that in the LWE problem, the noise function $\cR(v)$ corresponds to adding to each entry in $v$ a sample (called the ``error'') from a discrete Gaussian distribution, say of standard deviation $\sigma$, modulo $q$. The first search-to-decision reduction for LWE is from Regev's work that originally defined the problem \cite{DBLP:journals/jacm/Regev09}. This reduction preserves the error distribution, but requires the modulus $q$ to be prime. It is also not sample-preserving, with the search algorithm needing a $\tilde{O}(1/\eps^2)$-factor of extra samples, where $\eps$ is the advantage of the distinguisher. Further, the complexity\footnote{By ``complexity'', we typically mean the product of the running time and the inverse of the success probability.} of the reduction scales polynomially with the modulus $q$.

Several subsequent works have since extended this reduction to composite moduli $q$, under various conditions. Peikert~\cite{DBLP:conf/stoc/Peikert09} obtained a reduction for $q = p_1 \cdots p_k$ (for distinct primes $p_i$) that has runtime scaling polynomially with $(\sum_i p_i)$, under the requirement that the advantage of the distinguisher $\eps$ is larger than $\Omega(\sum_i \exponential(-p_i^2\sigma^2))$.  
Applebaum et al.~\cite{DBLP:conf/crypto/ApplebaumCPS09} obtained a reduction for $q = p^e$ (for prime $p$) that has runtime scaling polynomially with $p$ and $e$. Micciancio and Peikert~\cite{DBLP:conf/eurocrypt/MicciancioP12} combined both results and obtained a reduction for $q = p_1^{e_1} \cdots p_k^{e_k}$ (for distinct primes $p_i$) that has runtime scaling polynomially with $\sum (p_i + e_i)$ and a similar requirement on the distinguisher's advantage: $\eps > \Omega(\sum_i \exponential(-p_i^{2e_i}\sigma^2/q^2))$. All of these reductions also have a $\tilde{O}(1/\eps^2)$-factor increase in the number of samples needed by the search algorithm.

Later work by Brakerski et al.~\cite{DBLP:conf/stoc/BrakerskiLPRS13} showed how to reduce Decision LWE over a modulus $q'$ to Decision LWE over any modulus $q \leq q'$, with runtime scaling only polylogarithmically in $q'$. Combined with the other reductions mentioned above, this enables reductions from Search LWE over a nice $q'$ (that has efficient search-to-decision reductions) to Decision LWE over any $q \leq q'$. 
\[
    \text{Search-LWE}_{m^\prime, q^\prime, \sigma^\prime} \implies \text{Decision-LWE}_{m^\prime, q^\prime, \sigma^\prime} \implies \text{Decision-LWE}_{m, q, \sigma}
\]
The price of Brakersi et al.'s reduction is that the noise-to-modulus ratio of the resulting instance is larger: $\sigma/q = \Omega(\sqrt{n}\cdot \sigma'/q')$. The number of variables remains the same, while the number of samples decreases by $\Omega(n + \log \log q)$ (i.e., $m = m^\prime - \Omega(n + \log \log q)$).

Most relevant to our work is the search-to-decision reduction of Micciancio and Mol \cite{DBLP:conf/crypto/MicciancioM11}, which is the only known sample-preserving classical reduction for LWE.\footnote{The sample-preserving reduction of Applebaum et al. for LPN~\cite{DBLP:journals/joc/ApplebaumIK09} is based on similar ideas.}  Their reduction avoids the dependence on $q$ and has runtime scaling polynomially with the noise magnitude $\sigma$ instead.  In the special case when $q$ is of the form $p^e$ and the noise is restricted to be at most $(p/2)$ in each element, the complexity reduces to linear in $p$ (and polylogarithmic in $q$).  More recently, Sudo et al.~\cite{DBLP:journals/ieiceta/SudoHTY25} showed a quantum sample-preserving search-to-decision reduction for LWE that uses only two queries to the distinguisher, but succeeds only with probability inverse polynomial in $q$.

Using techniques related to those of Micciancio and Mol, Bogdanov et al. \cite{DBLP:conf/tcc/BogdanovGMRR16} effectively showed the only known search-to-decision reduction for LWR, which has complexity polynomial in the modulus $q$.\footnote{The focus of their work is different, but this result is implicit in the proof of their Lemmas 3 and 4.} 

\begin{table}[t]
\centering
\begin{tabular}{llllll}
\toprule
\textbf{Work} & \textbf{Problem} & \textbf{Moduli} & \textbf{Noise} & \textbf{Samples} & \textbf{Complexity} \\
\midrule
\cite{DBLP:journals/jacm/Regev09} & LWE & Prime & Preserving & $\tilde{O}(m/\eps^2)$ & $\poly(q)$ \\
\cite{DBLP:conf/eurocrypt/MicciancioP12} & LWE & $p_1^{e_1} \cdots p_{k}^{e_k}$ & Preserving & $\tilde{O}(m/\eps^2)$ & $\poly(\sum_{i} p_i + e_i)$ \\
\cite{DBLP:conf/stoc/BrakerskiLPRS13}+\cite{DBLP:conf/eurocrypt/MicciancioP12} & LWE & All & $O(\sqrt{n})$ loss  & $\tilde{O}(m/\eps^2)$ & $\poly(\sum_{i} p_i' + e_i')$ \\
\cite{DBLP:conf/crypto/MicciancioM11} & LWE & Prime Power & Preserving & Preserving & $\poly(\sigma)$ \\
\cite{DBLP:conf/stoc/BrakerskiLPRS13}+\cite{DBLP:conf/crypto/MicciancioM11} & LWE & All & $O(\sqrt{n})$ loss  & $m + \tilde{O}(n)$ & $\poly(\log q, \sigma)$ \\
\cite{DBLP:conf/tcc/BogdanovGMRR16} & LWR & Prime & Preserving & Preserving & $\poly(q)$ \\
\textbf{This Work} & LWE/LWR & Prime & Preserving & Preserving & $\poly(\log q, q/\sigma)$ \\
\bottomrule
\end{tabular}
\caption{Comparison of search-to-decision reductions. $n$ is the dimension of the secret, $m$ is the number of samples, $\sigma$ is the noise magnitude (ratio of moduli in case of LWR), $q$ is the modulus, and $\eps$ is the advantage of the decision algorithm. Complexity omits $\poly(m, n, 1/\eps)$ factors. For reductions involving \cite{DBLP:conf/stoc/BrakerskiLPRS13}, the reduction is from Search LWE with a possibly different modulus $q' = p_1'^{e_1'} \cdots p_{k}'^{e_k'}$.}
\label{tab:past_reductions}
\end{table}

\paragraph{Limitations of Existing Reductions.} The easy reduction from decision to search for Noisy Linear Equations preserves all parameters: secret dimension, number of samples, error magnitude, and modulus. To establish \emph{equivalence} of search and decision, so should the reverse reduction.  As described above, all currently known reductions either fail to preserve one or more of these, or else have running time that is polynomial in the modulus or error magnitude (and hence are not efficient when these are large). The features of these reductions are summarised in \cref{tab:past_reductions}.

\subsection{Our Results}

We construct a new parameter-preserving search-to-decision reduction for the Noisy Linear Equations problem for noise functions that are robust to $\ell_\infty$ perturbations (see below for details). Where it applies, our reduction is sample-preserving and, roughly speaking, its complexity scales with the noise-to-modulus ratio rather than with the magnitude of the noise or the modulus. Among other corollaries (see \cref{sec:corollaries}), our results imply the following for any prime modulus $q$:
\begin{itemize}
    \item Search-to-decision reduction for LWE that preserves the dimension of the secret $n$, number of samples $m$, modulus $q$, and error magnitude $\sigma$, with running time that is $\poly(n,m,\log{q},q/\sigma)$
    \item Search-to-decision reduction for LWR that preserves the dimension of the secret $n$, number of samples $m$, modulus $q$, and rounding modulus $p$, with running time that is $\poly(n,m,\log{q},p)$
\end{itemize}
In particular, our reduction for LWE runs in polynomial time even if the modulus $q$ is exponentially large, as long as the error rate $\sigma/q$ is inverse-polynomial. 

\medskip
More generally, our reduction applies to Noisy Linear Equations problems with a noise function $\cR$ that can tolerate small perturbations. We say $\cR$ is \emph{$B$-smooth} if for every $v$ of $\ell_\infty$-norm at most $B$, with sufficiently high probability over the choice of $A$, the distribution of $\cR(Ax+v)$ is very close to that of $\cR(Ax)$. 
LWE is smooth when $B$ is somewhat smaller than the error magnitude $\sigma$. LWR is smooth when the $B$ is somewhat smaller than the ratio of the moduli $q/p$. (See \cref{sec:prelims} for the relevant formal definitions and statements.)

\begin{theorem}[Informal, see \cref{theorem:main} and \cref{sec:corollaries}]
    Consider any $n,m,p,q \in \mathbb{Z}^+$, with $q$ being prime, and a $B$-smooth noise function $\cR: \Integer{q}^{m} \to \Integer{p}^{m}$. For the \emph{Noisy Linear Equations problem} over $\Integer{q}$ with sample count $m$, secret dimension $n$, and noise function $\cR$, there exists an algorithm that, given a decision algorithm that runs in time $T$ and has advantage $\eps$, solves the search problem in time $\poly(n, m, \log q, q/B\eps, T)$ with success probability $\tilde{\Omega}(B\eps/q)$.
\end{theorem}

Beyond prime moduli, our reductions also work as above for moduli $q$ that do not have any small prime factors (see \cref{sec:corollaries} for details). Our reductions also work for general moduli $q$ in some special cases, such as the entries in the secret vector being very small. 

Otherwise, for general $q$, our search algorithm instead produces a list of size $\tilde{O}((B/q\eps)^2)$ of candidate vectors that contains a vector $c \cdot x$ for some $0 < |c| \leq \tilde{O}((B/q\eps)^2)$, where $x$ is the secret vector. If $c$ turns out to be coprime with $q$, then by trial and error over all elements in the list and all $c$, we would be able to recover the secret if equipped with an algorithm to verify the solution. 

When $c$ is not coprime to $q$, the output of our reduction reveals only partial information about the secret. This is, in fact, a limitation of all known parameter-preserving reductions for the following reason.
When $q$ is not prime, an LWE distinguisher could censor substantial information about the secret.  For example, if the modulus $q$ is even, the distinguisher might base its decision only on the value of the instance reduced modulo $q/2$.  The output of the distinguisher would then be identical for every pair of secrets congruent modulo $q/2$.  All known parameter-preserving reductions, including ours, either leave the secret intact or modify it by an affine transformation.  Such reductions can therefore at best recover the secret up to congruencies modulo small factors of $q$.  Bogdanov et al.~\cite[Lemma 4]{DBLP:conf/tcc/BogdanovGMRR16} point out a related limitation in a step of their reduction from search LWE to decision LWR.

\paragraph{Open Problems.} Neither our algorithm nor that of Micciancio and Mol~\cite{DBLP:conf/crypto/MicciancioM11} runs in time $\poly(n)$ in a parameter setting where the error magnitude is both $\omega(\poly(n))$ and $o(q/\poly(n))$. Ideally, we would like to have a parameter-preserving reduction whose complexity scales only with $\polylog(q)$ for all noise parameter settings. We do not see a fundamental barrier to such a reduction, but it seems out of the reach of existing techniques. The other open problem is to perform exact secret recovery reduction when $q$ is an arbitrary composite modulus, without any restrictions on the secret $x$.

\subsection{Technical Overview}

\label{sec:technical_overview}

In this subsection, we first analyse why most existing reductions have a complexity that scales linearly with the size of the modulus. Then we discuss how \cite{DBLP:conf/crypto/MicciancioM11} achieves a sample-preserving reduction that scales with the size of the error. Finally, we introduce our reduction technique.

\paragraph{Prior Reductions.} Regev's initial reduction for LWE \cite{DBLP:journals/jacm/Regev09} uses a coordinate-guessing technique to recover the secret. Make a guess $g\in\Integer{q}$ for the first co-ordinate of secret $x_1$. For each sample $(a,\ip{a,x} + e)$, it is modified to $(a + (b, 0, \dots, 0),\ip{a,x} + e + b \cdot g)$ where $b$ is randomly sampled from $\Integer{q}$. If $g = x_1$, then the distribution stays the same. Otherwise, it becomes random. Because the reduction has to iterate through all possible guesses, the reduction complexity grows with the size of the domain. Later works by \cite{DBLP:conf/stoc/Peikert09, DBLP:conf/crypto/ApplebaumCPS09, DBLP:conf/eurocrypt/MicciancioP12} extend this technique to support larger modulus $q = p_1^{e_1} \cdots p_k^{e_k}$. The rough idea is to use the Chinese Remainder Theorem and guess the value of the secret modulo the small primes. Therefore, the complexity of the reduction has a polynomial dependence on $\sum_i (p_i + e_i)$.

Micciancio and Mol obtain a sample-preserving reduction \cite{DBLP:conf/crypto/MicciancioM11} using an approach analogous to the sample-preserving reduction for LPN with $q = 2$ by Applebaum, Ishai and Kushilevitz \cite{DBLP:journals/joc/ApplebaumIK09}. The main idea is to use the distinguisher to construct a weak inner product predictor $\algP(y)$ that ``predicts'' $\ip{y,x} \text{ mod } q$, where $y \in \Integer{q}^n$ is an input and $x$ is the secret. This can be done by transforming $(A, Ax + e)$ to $(A + ry^\top, Ax + e + rg)$ where $r$ is a random vector from $\Integer{q}^{m}$ and $g \in \Field{q}$ is a guess for the value $\ip{y,x}$. If $g = y^\top x$, then the distribution stays the same. Otherwise, it becomes random. The predictor $\algP$ performs this transformation, runs the LWE distinguisher on the resulting instance, and outputs whatever it outputs. They show that this weak predictor $\algP(y)$ has significant Fourier weight at $x$, and use the learning algorithm from Akavia, Goldwasser and Safra \cite{AGS03} to recover its heavy Fourier coefficients, and thus the secret.

Making a guess for the value of $\ip{y, x} \bmod{q}$ degrades the advantage by $1/q$, which causes the overall learning algorithm to have a running time polynomial in $q$. To avoid this polynomial dependency on the modulus, they use the equivalence of LWE and Knapsack and instead construct a predictor for $\ip{y, e}$, where $e$ is the error vector. The domain of the predictor $\algP(y)$ is changed to $y \in \Integer{p}^m$, where $p$ is a prime and $p \geq \max e_i$, the output is then $\algP(y) = \ip{y, e} \text{ mod } p$. This prime $p$ scales with the error size and not with the original modulus $q$. When each term of the $y_i$ is at most $p$, the predictor guessing range for $g = \ip{y, e}$ (without modulo $p$) is reduced to $p \cdot m \cdot \max_i e_i$. They will make a guess $g \in [0, p \cdot m \cdot \max_i e_i]$ for the value of $\ip{y, e}$ without modulo $p$. The loss in advantage is then a factor of $1/\poly(p, m)$. Running the learning algorithm on such a predictor has running time polynomial in $\poly(p, \log q)$.

Ultimately, existing known reductions involve iterating through all possible guesses for some secret value, which incurs a reduction complexity that scales linearly with the size of the domain of this value ($q$ in the case of \cite{DBLP:journals/jacm/Regev09}, $|e|_\infty$ in the case of \cite{DBLP:conf/crypto/MicciancioM11}, etc.). If iteration through the domain is used in the reduction, it will be difficult to obtain a reduction that runs in time polylogarithmic in the size of the modulus, secret and error.  

\paragraph{Threshold Distinguisher.} 
To obtain a better dependence on $q$, we would need to use the single-bit output of the distinguisher more effectively than to rule out just one possible guess each time. Ideally, we hope that a single query to the distinguisher would split the search space in half, similar to a binary search. 

Consider a distinguishing algorithm $\algD$ for the Noisy Linear Equation problem that takes as input $(A,\cR(Ax))$ for some $A\in\Integer{q}^{m\times n}$ and $x\in\Integer{q}^n$, and outputs a bit from $\{-1,1\}$. We use this to construct an algorithm $\algP$ that takes as input $(A,\cR(Ax))$ and also a vector $y \in \Integer{q}^n$, and somehow signals whether $\abs{\ip{y, x}_q} \leq B$ or $\abs{\ip{y, x}_q} > B$, for some parameter $B$. The closer $B$ is to $q/2$, the better, because that reveals more information about inner produce $\ip{y,x}$, and thus the secret $x$. 

\medskip
To see how the construction is done, consider taking the input $(A, \cR(Ax))$ and adding $vy^\top \in \Integer{q}^{m \times n}$ to the matrix, where $v \gets \{0,1\}^m$. In other words, the input becomes
\[
    (A + vy^\top, \cR(Ax))
\]
Since $A$ is random, this distribution is essentially
\[
    (A + vy^\top, \cR(Ax)) \equiv (A, \cR(Ax - vy^{\top}x))
\]
If $y^{\top}x = \abs{\ip{y, x}_q} \leq B$ is small, this is like adding a small vector to the already noisy $\cR(Ax)$. If the noise function $\cR$ absorbs it, then the distribution is close to the noisy linear equations distribution. 
\[
    (A, \cR(Ax - vy^{\top}x)) \approx (A, \cR(Ax))
\]
Otherwise, if $y^{\top}x = \abs{\ip{y, x}_q}$ is large, it overwhelms the noise. And if we add many such $vy^\top$ to $A$, it essentially becomes randomised and independent of $\cR(Ax)$. 

\medskip
Unfortunately, we cannot afford to set $B = q/2$ because it will overwhelm the noise in most natural noise distributions. The actual $B$ that we use is the largest $B$ such that for any vector $z$ whose elements have magnitude at most $B$, the statistical distance between $\cR(Ax+z)$ and $\cR(Ax)$ is bounded by $O(\eps)$. 

In the proof, starting with the observation about adding $vy^\top$ mentioned above, we will use a hybrid argument to construct the algorithm $\algP$. Informally, what we will show is the following (\cref{sec:threshold_distinguisher}):
\[
    \left|
    \begin{aligned}
        & \Pr[\algP(y, A, \cR(Ax)) = 1 | \abs{\ip{y, x}_q} \leq B] \\
        &\quad - \Pr[\algP(y, A, \cR(Ax)) = 1 | \abs{\ip{y, x}_q} > B]
    \end{aligned}
    \right| \geq \tilde{\Omega}(\eps)
\]
Where $y \gets \Integer{q}^n$, $A \gets \Integer{q}^{m \times n}$, $x \gets \Integer{q}^n$

\paragraph{Heavy Fourier Coefficients.} 
Just having $\algP$ is still not enough to recover the secret $x$. As we would like to have a sample-preserving reduction, we use an idea similar to \cite{DBLP:conf/crypto/MicciancioM11, DBLP:conf/tcc/BogdanovGMRR16}, where the information about $x$ is obtained by looking into the Fourier coefficients of $\algP$. We show that for a good fixed input $(A, x, \cR(Ax))$, the function $\algP(y) = \algP(y, A, \cR(Ax))$ has a heavy Fourier coefficient at some character $c \cdot x$, where $c$ is a small integer and $x$ is the secret of the noisy linear equations. We can then recover a list containing $c \cdot x$ with  the algorithm of Akavia, Goldwasser, and Safra \cite{AGS03} (\cref{sec:learn_fourier_coefficients}).

To analyse the Fourier spectrum of $\algP$, we use an intermediate threshold function $S$ to isolate the important Fourier characters (\cref{def:threshold_function}). The function $S$ is defined as 
    \[
        S(y) = \begin{cases}
            -1 &\text{ if } \abs{\ip{y, x}_q} \leq B\\
            1 &\text{otherwise}
        \end{cases}
    \]
with $y$ being the input vector, and $x$ being the secret of the noisy linear equations. Because of the properties of $\algP$ described earlier, we can show that $\algP$ is correlated with $S$ (\cref{sec:corr_threshold}). 
\[
    \exp{(\algP(y) - \mu_P)(S(y) - \mu_S)} \geq \tilde{\Omega}(\dfrac{B\eps}{q})
\]
where $\mu_P = \exp{P}$ and $\mu_S = \exp{S}$. This correlation implies that $P$ must have some non-negligible Fourier weight on the heavy Fourier characters of $S$.

By definition of $S$, it only has non-zero Fourier coefficients at multiples of $x$. Furthermore, we show that the Fourier weight of $S$ concentrates on the characters $c \cdot x$ where $c$ is small (\cref{sec:heavy_fourier_coeff}). Due to their correlation, this implies that $\algP$ also contains at least one heavy Fourier coefficient at a character $c \cdot x$ where $c$ is small.

\paragraph{Recovering the Secret} Finally, if $q$ is a prime or $c$ turns out to be coprime with $q$, we can go through every element in the output list and try every possible $c$. Since the list and $|c|$ are bounded by $\poly(q/B\eps)$, this process finishes in time that grows polynomially with $q/B\eps$. If a verification algorithm exists, we can check each answer candidate and return the correct one. Alternatively, randomly pick an element from the list and choose a random $c$, incurring some loss in the success probability.

\subsection{Paper Outline}
In \cref{sec:prelims}, we present the notation and technical tools used throughout the paper.
In \cref{sec:reduction}, we prove the main theorem.
In \cref{sec:corollaries}, we derive various corollaries of our main theorem, including the reductions for LWE and LWR.

\section{Preliminaries}
\label{sec:prelims}

\paragraph{Notation} 
We use $x_i$ to denote the $i$-th bit of the binary string (vector) $x$ and $A_i$ to denote the $i$-th row of the matrix $A$. Let $\Field{p} $ denote the order-$p$ prime field. Let $\Integer{q} $ denote the integer ring modulo $q$. Write $[n] := \{1,2,\dotsc, n\}$. Write $\round{x}$ as the number $x$ rounded to the nearest integer. Write $\ip{y, x}_q \in \Integer{q}$ as the inner product of vectors $x$ and $y$ modulo $q$. Define $\abs{x}$, where $x \in \Integer{q}$, as $\abs{x} = \min(x^\prime, x^\prime - q)$ where $x^\prime$ is $x$ lifted to the integer ring.

\paragraph{Distributions} Denote $D_1 \equiv D_2$ if the distributions $D_1$ and $D_2$ are the same. Let $U(D)$ be the uniform distribution on the set $D$.
The notation $x \gets D$ means that $x$ is sampled from the distribution $D$; when $D$ is a set (we abuse notation), it means $x$ is sampled from $U(D)$. 

\begin{definition}[Discrete Gaussian Distribution]
    Define $D_\sigma$ as the discrete Gaussian distribution over integers with parameter $\sigma > 0$ and centre at 0. 
    \[
        \rho_\sigma(x) = \exponential(-x^2/2\sigma^2)
    \]
    \[
        \Pr[D_{\sigma} = x] = \rho_\sigma(x)/\sum_{y \in \mathbb{Z}} \rho_{\sigma}(y)
    \]
\end{definition}

\begin{definition}[Noisy Linear Equations]
\label{def:noisy-linear-eq}
    Let $n,m,p, q \in \mathbb{Z}^+$. For sample count $m$, secret dimension $n$, and a noise function $\cR: \mathbb{Z}_q^m \to \mathbb{Z}_p^{m}$, algorithms $\algS$ and $\algD$ are said to have advantage $\eps \in [0, 1]$ in solving the \emph{Search} and \emph{Decision} Noisy Linear Equations problem, respectively, if they satisfy the following conditions:
    \begin{itemize}[itemsep=0pt]
        \item \emph{Search:} Suppose $A\gets \Integer{q}^{m \times n}$, $x \gets \Integer{q}^n$, $\algS$ recovers the secret $x$ with probability at least $\eps$:
        \[
            \Pr[\algS(A, \cR(Ax)) = x] \geq \eps
       . \]
        
        \item \emph{Decision:} Suppose $A\gets \Integer{q}^{m \times n}$, $x \gets \Integer{q}^n$, $b \gets \Integer{q}^m$, $\algD$ distinguishes samples from uniform with advantage at least $\eps$:
        \[
            \abs{\Pr[\algD(A, \cR(Ax)) = 1] - \Pr[\algD(A, \cR(b)) = 1]} \geq \eps.
        \]
    \end{itemize}
\end{definition}

\begin{definition}[Learning With Error (LWE)]
    The \emph{Learning With Error (LWE)} problem with noise parameter $\sigma$ is a \emph{Noisy Linear Equations} problem with the noise function $\cR: \Integer{q}^m \to \Integer{q}^m$ defined as $\cR(y) = (y + e) \text{ \emph{mod} } q$, where $e \gets D_\sigma^m$.
\end{definition}

\begin{definition}[Learning With Rounding (LWR)]
    The \emph{Learning With Rounding (LWR)} problem with modulus $p < q$ is a \emph{Noisy Linear Equations} problem with the noise function $\cR: \Integer{q}^m \to \Integer{p}^m$ defined, for each $i\in[m]$, as $\cR(y)_i = \round{\dfrac{p}{q}y_i^\prime} \text{ \emph{mod} } p$, where $y_i^\prime$ is $y_i$ lifted to the reals.
\end{definition}

\subsection{Coprime secret}

Throughout the proof, we will only consider secrets of the noisy linear equations (that is, the vector $x$ in \cref{def:noisy-linear-eq}) that satisfy $\gcd(x, q) = \gcd(x_1, \dots, x_n, q) = 1$.
This is a mild assumption when we sample $x$ at random because the event of obtaining such a secret occurs with very high probability. We can precisely calculate the probability of this event using the Chinese Remainder Theorem.
\[
    \Pr[\gcd(x_1, \dots, x_n, q) = 1] = \prod_{\substack{p | q\\p \text{ is } prime}} (1 - \dfrac{1}{p^n}) \geq (1 - \dfrac{1}{2^n})^{\log q} \geq 1 - \dfrac{\log q}{2^n}
\]
which is a value that is exponentially close to $1$.

\begin{definition}
    Define $\cP(\Integer{q}^n) = \{x \in \mathbb{Z}_q^n \mid \gcd(x_1, \dots, x_n, q) = 1\}$
\end{definition}

\begin{proposition}
    \label{prop:coprime_secret} For any $q \in \Integer{\geq 2}$, 
    $TV(U(\cP(\Integer{q}^n)), U(\Integer{q}^n))  \leq \dfrac{\log q}{2^n}$
\end{proposition}

\subsection{Noise Tolerance}

In this subsection, we show that both LWE noise and LWR noise can absorb small deviations in the noisy codeword (which is the term we shall use to refer to $Ax$). Specifically, we will be adding a vector $v \gets \{0, B\}^m$ to the noisy codeword of LWE and LWR, where $B$ is something much smaller than the noise magnitude. We show that the total variation distance to the original distribution will still be close despite adding $v$. This property is essential for our search-to-decision reduction.

\begin{definition}[Noise Tolerance]
    \label{def:noise_tolerance}
    Let $n,m,p,q, B \in \mathbb{Z}^+$ and $\eps \in [0, 1]$. Let $\cR: \mathbb{Z}_q^m \to \mathbb{Z}_p^{m}$ be a noise function. We say that $\cR$ is \emph{$(B, \eps)$-noise tolerant} if for every $\abs{b} \leq B$ and every $v \in \{0, b\}^m$,
\[
    TV\Big( (A, \cR(Ax)), (A, \cR(Ax + v)) \Big) \leq \eps
\]
where $A \gets \mathbb{Z}_q^{m \times n}$ and $x \gets \cP(\Integer{q}^n)$.

Let $\cB^\prime(\cR,\eps) \in \Integer{}$ be the largest $B$ such that such that $\cR$ is $(B, \varepsilon)$-noise tolerant. The \emph{tolerance bound function} is defined as $\cB(\cR, \eps) = \min(\cB^\prime(\cR,\eps), \floor{q/4 - 1})$.
\end{definition}

\begin{proposition}[LWE Noise Tolerance]
    \label{prop:LWE_noise_flooding}
    Let $n,m,q,B \in \mathbb{Z}^+$, $\sigma \in \Reals$, and $\cR$ be the LWE noise function with noise parameter $\sigma$. Suppose $A \in \Integer{q}^{m \times n}$, $x \in \Integer{q}^n$, $v \in \{0, b\}^m$ where $\abs{b} \leq B$, then
    \[
     TV\Big((A, \cR(Ax)), (A, \cR(Ax + v))\Big) \leq \dfrac{B\sqrt{m}}{2\sigma}
    \]
    Which implies $\cB(\cR, \eps) \geq \min(2\eps\sigma/\sqrt{m}, \floor{q/4 - 1})$.
\end{proposition}
\begin{proof}
    Expanding the noise function $\cR$, then by data processing inequality
    \begin{align*}
         TV\Big((A, \cR(Ax)), (A, \cR(Ax + v))\Big) &= TV\Big((A, Ax + e), (A, Ax + e + v)\Big)\\
         &\leq TV(e, e+v)
    \end{align*}
    Where $e \gets D_\sigma^m$ and addition is done over $\Integer{}$.
    
    Without loss of generality, assume that $b \geq 0$. Looking at a single dimension $i \in [m]$ such that $v_i = b$, and write down the KL divergence
    \begin{align*}
        D_{KL}(b + e_i\ \|\ e_i)  &= \sum_{z \in \Integer{}} \Pr[b + e_i = z] \log \dfrac{\Pr[b + e_i = z]}{\Pr[e_i = z]}\\
        &=  \sum_{z \in \Integer{}} \Pr[b + e_i = z] \log \exponential(\dfrac{(z)^2}{2\sigma^2} - \dfrac{(z-b)^2}{2\sigma^2})\\
        &= \sum_{z \in \Integer{}} \Pr[b + e_i = z]  {\dfrac{2bz -b^2}{2\sigma^2}}
    \end{align*}    
    Let $Y = b + e_i$, and simplify the statement above
    \[
        \sum_{z} \Pr[Y = z]  {\dfrac{2bz -b^2}{2\sigma^2}} = \expec{Y \sim b + D_{\sigma}}{\dfrac{2bY - b^2}{2\sigma^2}} = \dfrac{2b\expec{}{Y} - b^2}{2\sigma^2} = \dfrac{b^2}{2\sigma^2} \leq \dfrac{B^2}{2\sigma^2}
    \]
    For $m$ dimensions, Pinsker's Inequality bounds the total variation by
    \[
        TV(e + v, e)  \leq \sqrt{\dfrac{1}{2} D_{KL}(e + v \ \|\ e)} = \sqrt{\dfrac{m}{2} D_{KL}(e_i + v_i \ \|\ e_i)} \leq \dfrac{B\sqrt{m}}{2\sigma}.
    \hfill\qedhere \]
\manualqed\end{proof}

\begin{proposition} [LWR Noise Tolerance]
    \label{prop:LWR_noise_resilient}
    Let $n,m,p,q \in \mathbb{Z}^+$, $p < q$ and $\cR$ be the LWR noise function with rounding modulus $p$. Suppose $A \gets \Integer{q}^{m \times n}$, $x \in \cP(\Integer{q}^n)$, $v \in \{0, b\}^m$ where $\abs{b} \leq B < q/p$, then
    \[
     TV\Big((A, \cR(Ax)), (A, \cR(Ax + v))\Big) \leq \dfrac{pBm}{q}
    \]
    Which implies $\cB(\cR, \eps) \geq \min(q\eps/pm, \floor{q/4 - 1})$.
\end{proposition}
\begin{proof}
    Without loss of generality, assume that $b \geq 0$. Let $E$ be the event that $\cR(Ax) = \cR(Ax + v)$. This event happens when every value in $Ax$ is far away from the boundary where rounding occurs. In other words, for each coordinate value $(Ax)_i$, lifted to the reals $(Ax)_i^\prime$, we have that $\round{\dfrac{p}{q}(Ax)_i^\prime} \equiv \round{\dfrac{p}{q}((Ax)_i^\prime + b)}$ mod $p$. The total variation distance can then be upper-bounded using this event.
    \[
        TV\Big((A, \cR(Ax)), (A, \cR(Ax + v))\Big) \leq \Pr[\bar{E}]
    \]
    As $\gcd(x, q) = 1$ and $A$ is randomly sampled, $Ax$ is distributed just as a random vector over $\Integer{q}^m$. For a list of elements in $\Integer{q}$ that round to the same element in $\Integer{p}$, only the last $b$ elements in the list would round to a different element after adding the value $b$, while the rest stay the same. 
    
    This means that at most $pB$ elements in $\Integer{q}$ will change value after addition. So for a random element in $\Integer{q}$ the probability that it falls into this boundary region is at most $pB/q$. By union bound, the probability of $\bar{E}$ is upper bounded.
    \[
        TV\Big((A, \cR(Ax)), (A, \cR(Ax + v))\Big) \leq \Pr[\bar{E}] \leq \dfrac{pBm}{q}.
    \hfill\qedhere \]
\manualqed\end{proof}

\subsection{Fourier Analysis}

Here we describe the basics of Fourier Analysis and establish the choice of basis in the paper. Denote $\omega_q = e^{-2\pi i/q}$.

\begin{definition}[Fourier Transform]
\label{def:ft}
    Let $n,q \in \Integer{}^+$. The Fourier transform of a function $S: \Field{q}^{n} \to \mathbb{C}$ is another function $\hat{S}:\Field{q}^n\to\mathbb{C}$ defined as
    \[
        \hat{S}(\chi) = \expec{y \gets \Integer{q}^m}{S(y)\omega_q^{\ip{\chi, y}}}
    \]
\end{definition}

\begin{proposition}[Inverse Fourier Transform] 
\label{prop:inverse-ft}
Let $\hat{S}$ be the Fourier transform of $S:\Field{q}^n\to\mathbb{C}$. Then
\[
    S(y) = \sum_{\chi \in \Integer{q}^m} \hat{S}(\chi) \omega_q^{-\ip{\chi, y}}
\]
\end{proposition}    

\begin{proposition}[Parseval–Plancherel]
    \label{prop:parseval}
    Let $n,q \in \Integer{}^+$, functions $P, S : \Field{q}^m \to \mathbb{C}$
    \begin{align*}
        \expec{y \gets \Integer{q}^n}{P(y)\overline{S(y)}} &= \expec{y \gets \Integer{q}^n}{(\sum_{\chi \in \Integer{q}^m} \hat{P}(\chi)\omega_q^{-\ip{\chi,y}})(\sum_{\chi \in \Integer{q}^m} \overline{\hat{S}(\chi)}\omega_q^{\ip{\chi,y}})} \\
        &= \sum_{\chi \in \Integer{q}^m}\hat{P}(\chi)\overline{\hat{S}(\chi)}
    \end{align*}
\end{proposition}

\noindent In particular, this implies the following for any such functions $P$
\[
\expec{y \gets \Integer{q}^n}{\abs{P(y)}^2} = \sum_{\chi \in \Integer{q}^m}\abs{\hat{P}(\chi)}^2
\]
\section{Reduction}
\label{sec:reduction}
This section describes the proof of the main theorem. 

\begin{theorem}
    \label{theorem:main}
     Consider any $n,m,p,q \in \mathbb{Z}^+$, a noise function $\cR: \Integer{q}^{m} \to \Integer{p}^{m}$, and any $\eps \in (0, 1]$. Let $t = \ceil{4\log (qm/\eps)}$, denote by $B$ the noise tolerance bound $\cB(\cR, \eps/4t)$, and set $\eps^\prime = (B\eps/2qt)$.
     
     Suppose there exists an algorithm for the \emph{Decision Noisy Linear Equations problem} over $\Integer{q}$ with sample count $m$, secret dimension $n$, and noise function $\cR$, and that the algorithm runs in time $T$ and has advantage $\eps$. 
     
     Then there exists an algorithm as follows. Given an input $(A, \cR(Ax))$, where $A \gets \Integer{q}^{m \times n}$ and $x \gets \cP(\Integer{q}^n)$, it outputs a list of size $\tilde{O}(1/(\eps^\prime)^2)$ that contains $c \cdot x$ for some $c \in \Integer{q} \setminus \{0\}$ such that $|c| \leq \ceil{32/(\eps^\prime)^2}$. This algorithm runs in time $\poly(n, m, \log q, 1/\eps^\prime, T)$ and successfully outputs such a list with probability $\Omega(\eps^\prime)$.
\end{theorem}

Proof overview is described in \cref{sec:technical_overview}

\paragraph{Section Outline} In \cref{sec:threshold_distinguisher} and \cref{sec:terminal_hybrid}, we will talk about the threshold distinguisher. \cref{sec:corr_threshold} is on the correlation of $\algP$ and $S$. \cref{sec:heavy_fourier_coeff} and \cref{sec:fourier_analysis_f} are on the concentration of the Fourier weight of $S$. Finally, in \cref{sec:learn_fourier_coefficients} we state the algorithm from \cite{AGS03} to learn heavy Fourier coefficients.

\subsection{Threshold Distinguisher}
\label{sec:threshold_distinguisher}

In this subsection, we describe the construction of a threshold distinguisher $\algP$ that can distinguish whether $\abs{\ip{y, x}_q} \leq B$ or $\abs{\ip{y, x}_q} > B$. Note that the inner product is performed over $\Integer{q}$.

\begin{lemma}
    \label{lemma:dist_adv}
    Let $n,m,p, q \in \mathbb{Z}^+$, $\eps \in (0, 1]$, a noise function $\cR: \Integer{q}^{m} \to \Integer{p}^{m}$, $t = \ceil{4\log (qm/\eps)}, B = \cB(\cR, \eps/4t)$. 
    Suppose there exists an algorithm $\algD$ for the \emph{Decision Noisy Linear Equations problem} over $\Integer{q}$ with sample count $m$, secret dimension $n$, noise function $\cR$, that has advantage $\eps$. 
    
    Then there exists an algorithm $\algP$ that outputs a value in $\{-1,1\}$ such that 
    \[
        \left|
        \begin{aligned}
            &\Pr[\algP(y, A, \cR(Ax)) = 1 | \abs{\ip{y, x}_q} \leq B]\\
            &\quad - \Pr[\algP(y, A, \cR(Ax)) = 1 | \abs{\ip{y, x}_q} > B]
        \end{aligned}
        \right| \geq \eps/4t
    \]
    Where $y \gets \Integer{q}^n$, $A \gets \Integer{q}^{m \times n}$, $x \gets \cP(\Integer{q}^n)$.
\end{lemma}

\begin{proof}
    The construction of $\algP$ is as follows:
    \begin{pseudocode}
        {\bf Algorithm $\algP$:} On input $(y, A, b = \cR(Ax))$ and Decision algorithm $\algD$, \\
        \pseudoline Sample $k \gets [0, t - 1]$ and $v \gets \{0, 1\}^m$\\
        \pseudoline Sample matrices $V_k \gets \{0, 1\}^{m \times k}$ and $Y_k \gets \Integer{q}^{k \times n}$ \\
        \pseudoline[algo:oracle_call] Output $\algD(A + V_kY_k + vy^\top, b)$
    \end{pseudocode}
    
    We use the hybrid argument to prove the lemma. First, define the hybrid distribution $\cH_k$ over $\Integer{q}^{m \times n} \times \Integer{p}^m$ as
    \[
        \cH_k = (A + V_k Y_k, \cR(Ax))
    \]
    \begin{claim}
        \label{claim:hybrid_close}
        For any $k \in \Integer{\geq 0}$, let $H = (A + V_kY_k + vy^\top, \cR(Ax))$.
        \begin{itemize}[itemsep=0pt]
            \item If $\abs{\ip{y, x}_q} \leq B$, then $TV(H, \cH_k) \leq \eps/4t$
            \item If no condition is imposed, then $H \equiv \cH_{k+1}$
        \end{itemize}   
    \end{claim}
     \begin{proof}
         Suppose that $\abs{\ip{y, x}_q} \leq B$, rewrite
        \[
             H \equiv (A + V_kY_k + vy^\top, \cR(Ax)) \equiv (A + V_kY_k, \cR(Ax - vy^\top x))
        \]
        The vector $vy^\top x = v \ip{y, x}_q$ is uniformly random in $\{0, \ip{y, x}_q\}^m$. By the definition of noise tolerance (\cref{def:noise_tolerance}), and the fact that B = $\cB(\cR, \eps/4t)$,
        \[
            TV(H, \cH_k) \leq \eps/4t
        \]   
        On the other hand, if no condition is imposed on $y$, it is uniformly random and
        \[
            H \equiv (A + V_kY_k + vy^\top, b) \equiv (A + V_{k+1}Y_{k+1}, b) \equiv \cH_{k+1}. \hfill\qedhere
        \]
     \manualqed\end{proof}
    It is clear from the definition that $\cH_0$ is distributed as a Noisy Linear Equation instance. 
    We claim that once $t$ is large enough, the terminal hybrid $\cH_t$ is distributed as a random instance.
    \begin{restatable}{claim}{terminalhybrid}\emph{(Terminal hybrid).}
        \label{claim:terminal_hybrid}
        Let $t = \ceil{4\log (qm/\eps)}$, $A \gets \Integer{q}^{m \times n}, b \gets \Integer{q}^m$,
        \[
             TV\Big(\cH_t, (A, \cR(b))\Big) \leq \eps^2/2qm \ll \eps
        \]
    \end{restatable}
        
     We defer the proof of the claim to the next section (\cref{sec:terminal_hybrid}). With this claim, we can then use hybrid argument to conclude a $\eps/2t$ advantage in distinguishing between $\cH_k$ and $\cH_{k+1}$ where $k \gets [0, t-1]$.
    \begin{restatable}{proposition}{hybridargument}
        \label{prop:hybrid_argument}
        \[
            \abs{\prob{k \gets [0, t-1]}{D(\cH_{k}) = 1} - \prob{k \gets [0, t-1]}{D(\cH_{k+1}) = 1}} \geq \eps/2t
        \]
    \end{restatable}
    \begin{proof}
        Standard hybrid argument with \cref{claim:terminal_hybrid}. Details deferred to \cref{proof:hybrid_argument}
    \manualqed\end{proof}

    Let $H$ be the distribution sent to the decision algorithm $\algD$ in \cref{algo:oracle_call} of the algorithm $\algP$. Using \cref{claim:hybrid_close} and \cref{prop:hybrid_argument}, we obtain a distinguisher that behaves differently under the condition of $\abs{\ip{y, x}_q} \leq B$ and no condition. i.e.
    \[
        \abs{\prob{y \gets \Integer{q}^{n}}{D(H) = 1\ |\ \abs{\ip{y, x}_q} \leq B} - \prob{y \gets \Integer{q}^{n}}{D(H) = 1}} \geq \eps/2t - \eps/4t = \eps/4t
    \]
    Where the loss of $\eps/4t$ is from the total variation distance between $H$ and $\cH_{k}$ when $\abs{\ip{y, x}_q} \leq B$.

    The condition $\abs{\ip{y, x}_q} > B$ can be imposed on the second term at no cost.  Let $E$ be the event that $\abs{\ip{y, x}_q} \leq B$. By law of total probability,
    \begin{align*}
        &\abs{\prob{}{D(H) = 1|E} - \prob{}{D(H) = 1}} \\
        =& \Pr[\Bar{E}] \cdot  \abs{\prob{}{D(H) = 1|E} - \prob{}{D(H) = 1|\Bar{E}}}
    \end{align*}
        
    Which implies
    \[
        \abs{\prob{}{D(H) = 1|E} - \prob{}{D(H) = 1|\Bar{E}}} \geq \dfrac{\eps/4t}{\Pr[\Bar{E}]} \geq \eps/4t. \hfill\qedhere
    \]
\manualqed\end{proof}

\subsection{Terminal Hybrid}
\label{sec:terminal_hybrid}
Here, we show the proof of the claim that the terminal hybrid is close to the uniform distribution.

\terminalhybrid*
    \begin{proof}
        Start by writing 
        \[
            \cH_t = (A + VY, \cR(Ax)) \equiv (A, \cR(Ax - VYx))
        \]
        Where $A \gets \Integer{q}^{m \times n}$, $V \gets \{0,1\}^{m \times t}$, $Y \gets \Integer{q}^{t \times n}, x \gets \cP(\Integer{q})$.
        
        Since $Y$ is a random matrix in $\Integer{q}^{t \times n}$, and $\gcd(x, q) = 1$, $y = Yx$ is uniformly random in $\Integer{q}^m$, even given $x$.  As every row $v$ of $V$ is a random $\{0, 1\}$ vector of min-entropy $2^{-t}$, and dot product with $y$ is a universal hash, by the leftover hash lemma~\cite{ILL89}, $(y, \langle v, y\rangle)$ is $\tfrac12 \sqrt{q 2^{-t}}$-statistically close to $(y, r)$ for a uniformly random $r$.  By a hybrid argument, the statistical distance between
        \[ (y, \langle v_1, y\rangle, \dots, \langle v_m, y\rangle) \quad\text{and}\quad  (y, r_1, \dots, r_m), \]
        where $v_i$ are the rows of $V$, is at most $m$ times larger.  In particular, $Vy$ is $\tfrac12 m \sqrt{q 2^{-t}}$-close to uniform.  By our choice of parameters, this is at most 
        \[ \tfrac12 m \sqrt{q(\eps/qm)^4}\leq \eps^2/2qm \tag*{\qedhere}\]
    \manualqed\end{proof}

\subsection{Correlation With Threshold Function}

\label{sec:corr_threshold}

\cref{lemma:dist_adv} essentially shows that the algorithm $\algP$ is correlated with the indicator for the event $\abs{\ip{y, x}_q} \leq B$.  For analytic purposes it will be convenient to express this correlation in terms of the covariance of $\algP$ and a random variable defined by sampling the inputs to the following function at random.

\begin{definition}[Threshold Function]
    \label{def:threshold_function}
    Let $n, q, B \in \Integer{}^+$. Define $S: \Integer{q}^{n} \times \Integer{q}^{n} \to \{-1, 1\}$ as
    \[
        S(y, x) = \begin{cases}
            1 &\text{ if } \abs{\ip{y, x}_q} \leq B\\
            -1 &\text{otherwise}
        \end{cases}
    \]
\end{definition}

The conversion from difference in conditional expectations to the covariance of $\pm 1$-valued random variables is given by the following formula.
\begin{lemma}
    \label{lemma:correlation}
    For any $\{-1, 1\}$-valued random variables $X$ and $Y$,
    \[
        \exp{(X - \mu_X)(Y - \mu_Y)} = \dfrac{\Var(Y)}{2} \cdot (\exp{X|Y=1} - \exp{X|Y=-1}),
    \]
    where $\mu_X = \exp{X}$, $\mu_Y = \exp{Y}$.
\end{lemma}

\begin{proof}
    We can rewrite the variance of $Y$ in terms of probabilities 
    \begin{align*}
        \Var(Y) &= \exp{Y^2} - \exp{Y}^2\\
        &= 1 - \left((1 - \Pr[Y = -1]) - \Pr[Y = -1]\right)^2\\
        &= -4\Pr[Y = -1]^2 + 4\Pr[Y = -1]\\
        &= 4\cdot\Pr[Y = -1](1 - \Pr[Y = -1]) \\
        &= 4\Pr[Y = -1]\Pr[Y = 1]
    \end{align*}
    Expand $\exp{(X - \mu_X)(Y - \mu_Y)}$ using law of total probability,
    \begin{align*}
        &\exp{(X - \mu_X)(Y - \mu_Y)} \\
        = &\exp{XY} - \mu_X\mu_Y\\
        =& \exp{X| Y = 1}\Pr[Y = 1]  - \exp{X| Y = -1}\Pr[Y = -1]\\
        &- (\exp{X| Y = 1}\Pr[Y = 1] + \exp{X| Y = -1}\Pr[Y = -1])\mu_Y\\
        =& \exp{X| Y = 1}\Pr[Y = 1](1 - \mu_Y) + \exp{X| Y = -1}\Pr[Y = -1](-1 - \mu_Y)
    \end{align*}
    Substitute $\mu_Y = \Pr[Y = 1] - \Pr[Y = -1]$, then 
    \[ 1 - \mu_Y = 2\Pr[Y = -1] \text{ and } (-1 - \mu_Y) = -2\Pr[Y = 1] \]
    Finally with $\Var(Y) = 4\Pr[Y = -1]\Pr[Y = 1]$,
    \begin{align*}
         &\exp{XY} - \mu_X\mu_Y  \\
         = &2\exp{X| Y = 1}\Pr[Y = 1]\Pr[Y = -1] - 2\exp{X| Y = -1}\Pr[Y = -1]\Pr[Y = 1]\\
        = &\dfrac{\Var(Y)}{2}(\exp{X| Y = 1} - \exp{X| Y = -1}). \tag*{\qedhere}
    \end{align*}
\manualqed\end{proof}
Applying \cref{lemma:correlation} to $\algP$ and $S$ gives the following corollary.
\begin{corollary}
    \label{cor:correlation}
    Using the same parameter definitions as in \cref{lemma:dist_adv},
    \[
        \exp{(\algP(y, A, \cR(Ax)) - \mu_P)(S(y, x) - \mu_S)} \geq (\dfrac{\Var(S)}{2})(\eps/4t) \geq \dfrac{B\eps}{2qt}
    \]
    Where $A \gets \Integer{q}^{m \times n}$, $x \gets \cP(\Integer{q}^n)$, $y \gets \Integer{q}^n$.
\end{corollary}
\begin{proof}
    From \cref{lemma:dist_adv}, without loss of generality,
    \[
        \Pr[\algP(y, A, \cR(Ax)) = 1 | S(y,x) = 1] - \Pr[\algP(y, A, \cR(Ax)) = 1 |  S(y,x) = -1] \geq \eps/4t
    \]
    Then the statement follows from \cref{lemma:correlation}. The calculation of the variance of $S$ is shown as follows:
    \[
        \Var(S) = 4 \Pr[S = 1]\Pr[S = -1] = 4(\dfrac{2B + 1}{q})(1 - \dfrac{2B + 1}{q}) \geq \dfrac{4B}{q} \hfill\qedhere
    \]
    Where the last inequality is from $B \leq \floor{q/4 - 1}$
\manualqed\end{proof}

\subsection{Heavy Fourier Coefficients}

\label{sec:heavy_fourier_coeff}

To analyse the Fourier coefficients of $\algP$, it is more convenient to fix the input problem instance $(A, x, \cR(Ax))$ to be any good instance. We say an instance $(A, x, \cR(Ax))$ is $\eps$-good if 
\[
    \expec{y \gets \Integer{q}^n}{(\algP(y, A, \cR(Ax)) - \mu_P)(S(y, x) - \mu_S)} \geq \eps
\]
By definition, $(\algP - \mu_P)(S - \mu_S)$ takes value in $[-4, 4]$. By Markov's inequality (\cref{prop:good_fraction}), 
\begin{equation}
\label{eq:corrstd}
    \expec{y \gets \Integer{q}^n}{(\algP(y, A, \cR(Ax)) - \mu_P)(S(y, x) - \mu_S)} \geq \frac{\eps}{2}
\end{equation}
for a $\eps/8$ fraction of $A, x$, and $\mathcal{R}$ (if randomized).  We fix these choices from here on and view $\algP$ and $S$ as functions of $y$ only.

\[ \algP(y) = \algP(y, A, \cR(Ax)) \text{ and } S(y) = S(y, x) \]

\begin{restatable}{proposition}{goodfraction}
    \label{prop:good_fraction}
    If $\abs{X} \leq 4$ and $\exp{X} \geq \eps$, then for any random variable $Y$, we have $\exp{X | Y} \geq \eps/2$ with probability at least $\eps/8$ over $Y$.
\end{restatable}

\begin{proof}
    Let $E$ be the set of $Y$ for which $\exp{X | Y} \leq \eps/2$.  By the law of total expectation, 
    \[ \eps \leq \exp X \leq \exp{X | Y \not\in E}\Pr(Y \not\in E) + \exp{X | Y \in E} \leq 4\Pr(Y \not\in E) + \eps/2 \]
    so that $\Pr(Y \not\in E) \geq \eps/8$.
\manualqed\end{proof}

By Plancherel's formula (see \cref{prop:parseval}), \eqref{eq:corrstd} has an equivalent formulation in the Fourier domain:
\begin{equation}
\label{eq:corrfourier}
\sum_{\chi \neq 0^m} \hat{\algP}(\chi) \cdot \hat{S}(\chi) \geq \frac{\eps}{2}.
\end{equation}
($\hat{S}$ will turn out to be real so we omit complex conjugation.)  

The function $S(y)$ can be expressed as $2f(\ip{y, x}) - 1$, where $f\colon \Integer{q} \to \mathbb{C}$ is the univariate function $1(\abs{t} \leq B)$.  Making the substitution $t = \ip{y, x}$ in the Fourier transform of $f$
\[ f(t) = \sum_{c \in \Integer{q}} \hat{f}(c) \omega^{-ct} \]
gives the Fourier expansion
\[ S(y) = \sum_{c \in \Integer{q}} 2\hat{f}(c) \cdot \omega^{-c\ip{y, x}} - 1 = \sum_{\chi \text{ s.t. } \chi = c\cdot x} 2\hat{f}(c) \cdot \omega^{-\ip{y, \chi}} - 1 
\]

Comparing this to the inverse Fourier Transform (\cref{prop:inverse-ft}) of $S$, we have the following observations:
\begin{itemize}[itemsep=0pt]
    \item For any $\chi$ that is not of the form $c\cdot x$ for some $c\in\Field{q}$, the Fourier coefficients $S(\chi)$ is $0$. 
    \item When $c \neq 0$, $\hat{S}(c \cdot x) = 2\hat{f}(c)$. 
    \item $\hat{S}(0^n) = 2\hat{f}(0^n)-1$
\end{itemize}
Inequality~\eqref{eq:corrfourier} therefore simplifies to 
\begin{equation}
\label{eq:corrf2}
\sum_{c \neq 0} \hat{\algP}(c \cdot x)  \cdot \hat{f}(c) \geq \dfrac{\eps}{4}. 
\end{equation}
As $f$ is a threshold function, most of its Fourier mass is concentrated on low frequencies (see Figure~\ref{fig:f}).

\begin{figure}[ht]
\centering
\begin{tikzpicture}
\begin{axis}[
    width=12cm,
    height=7cm,
    xmin=-15.5,
    xmax=15.5,
    ymin=-.1,
    ymax=.35,
    xtick={-15,-10,-5,0,5,10,15}
]

\addplot+[
    ycomb,
    mark=*,
    mark size=1.5pt,
    samples at={-15,-14,...,15},
] {(1 + 2*cos(deg(2*pi*x/31))
      + 2*cos(deg(4*pi*x/31))
      + 2*cos(deg(6*pi*x/31))
      + 2*cos(deg(8*pi*x/31))) / 31};

\end{axis}
\end{tikzpicture}
\caption{Fourier spectrum of $f$ with $q = 31$ and $B = 4$.}
\label{fig:f}
\end{figure}
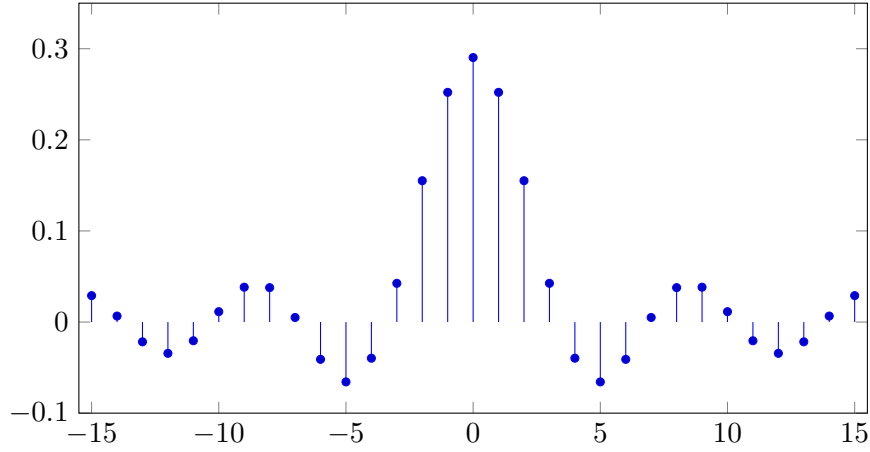

\begin{restatable}{claim}{lightfourier}
    \label{claim:fourier_conc}
    Let $\tau \in \Integer{}^+$ and define the light weight set $\cW_l = \{c \in \Integer{q}\ |\ |c| > \tau \}$, then
    \[
        \sum_{c \in \cW_l} \abs{\hat{f}(c)}^2 \leq 1/2\tau
    \]
\end{restatable}

\begin{restatable}{claim}{heavyfourier}
    \label{claim:fourier_conc_heavy}
    Let $\tau \in \Integer{}^+$ and define the heavy weight set $\cW_h = \{c \in \Integer{q}\ |\ 0 < |c| \leq \tau \}$, then
    \[
        \sum_{c \in \cW_h} \abs{\hat{f}(c)} \leq 1 + \ln \tau
    \]
\end{restatable}
Proofs of the two claims are deferred to the next section (\cref{sec:fourier_analysis_f}). We now use these to show that  \eqref{eq:corrf2} forces a heavy Fourier coefficient of the form $\hat{\algP}(c \cdot x)$.

\begin{lemma}
    \label{lemma:fourier_coeff} Assuming~\eqref{eq:corrstd}, there exists 
   $c \in \Integer{q}$ such that $c \neq 0$, $|c| \leq \ceil{32/\eps^2}$ and 
    \[
        \abs{\hat{\algP}(c \cdot x)} \geq \frac{\eps}{ 16\ln(14/\eps)}.
    \]
\end{lemma}

\begin{proof}
Set $\tau$ to $\ceil{32/\eps^2}$. As explained in the text above, \eqref{eq:corrstd} is equivalent to \eqref{eq:corrf2}. The sum in \eqref{eq:corrf2} splits as
    \[
        \sum_{c \in \cW_h}\hat{\algP}(c \cdot x)\hat{f}(c) + \sum_{c \in \cW_l}\hat{\algP}(c \cdot x)\hat{f}(c) \geq \eps/4. \]
    By the Cauchy-Schwarz inequality, Parseval's formula, and \cref{claim:fourier_conc}, 
    \begin{align*}
        \sum_{c \in \cW_l}\hat{\algP}(c \cdot x)\hat{f}(c) &\leq \sqrt{ \sum_{c \in \cW_l}\abs{\hat{\algP}(c \cdot x)}^2}\sqrt{ \sum_{c \in \cW_l}\abs{\hat{f}(c)}^2}
        \\
        &\leq \sqrt{\expec{y \gets \Integer{q}^n}{\algP(y)^2}}\sqrt{ \sum_{c \in \cW_l}\abs{\hat{f}(c)}^2} \leq \sqrt{\eps^2/64} = \eps/8
    \end{align*}
    Therefore, the sum over the heavy terms is lower bounded by 
    \[
        \sum_{c \in \cW_h}\hat{\algP}(c \cdot x) \cdot \hat{f}(c) \geq \eps / 8
    \] 
    By \cref{claim:fourier_conc_heavy}, the sum of all $|\hat{f}(c)|$ terms is at most $\ln \tau$.  Therefore at least one of the terms $\hat{P}(c \cdot x)$ must have magnitude at least
\[ \frac{\eps/8}{1 + \ln \tau} = \frac{\eps/8}{\ln \ceil{64e/\eps^2}} \geq \frac{\eps}{ 16\ln(14/\eps)} \]
by our choice of $\tau$.
\manualqed\end{proof}

\subsection{Fourier analysis of $f$}
\label{sec:fourier_analysis_f}

Next, we prove \cref{claim:fourier_conc,claim:fourier_conc_heavy}. We will make use of the following proposition.

\begin{proposition}
    \label{prop:fourier_coeff}
    For $c \in \Integer{q}$, $\abs{\hat{f}(c)} \leq \dfrac{1}{q|\sin(c\pi/q)|}$.
\end{proposition}

\begin{proof}
    The Fourier coefficients of $f$ are
    \[
        \hat{f}(c) = \expec{t \gets \Integer{q}}{\mathbf{1}(|t| \leq B) \cdot \omega_q^{c \cdot t}} = \dfrac{1}{q} \sum_{t = -B}^{B} \omega_q^{c \cdot t}
    \]
    Using the geometric progression formula,
    \[
    \dfrac{1}{q} \sum_{t = -B}^{B} \omega_q^{c \cdot t} = \dfrac{\omega_q^{-cB}(1 - \omega_q^{c(2B+1)})}{q(1 - \omega_q^c)} = \dfrac{\omega_q^{-cB}\omega_q^{c(2B+1)/2}}{q\omega_q^{c/2}} \cdot \dfrac{(\omega_q^{-c(2B+1)/2} - \omega_q^{c(2B+1)/2})}{(\omega_q^{-c/2} - \omega_q^{c/2})}
    \]
    Since $\omega_q^{-a} - \omega_q^{a} = -2i\sin(2\pi a/q)$, we can simplify to
    \[
        \dfrac{1}{q} \sum_{t = -B}^{B} \omega_q^{c \cdot t} = \dfrac{\sin(c(2B+1)\pi/q)}{q\sin(c\pi/q)}
    \]
    Since the denominator is always a positive number when $c \in \Integer{q}$  and the numerator is within range $[-1, 1]$, $|\hat{f}(c)| \leq 1/q|\sin(c\pi/q)|$.
\manualqed\end{proof}

\lightfourier*
\begin{proof}
    From \cref{prop:fourier_coeff},
    \begin{align*}
        \sum_{c \in \cW_l} \abs{\hat{f}(c)}^2 \leq \frac{1}{q^2} \cdot \sum_{c \in \cW_l} \dfrac{1}{\sin(c\pi/q)^2}
    \end{align*}
    As $\sin^2$ is an even function,
    \[
        \frac{1}{q^2} \sum_{c \in \cW_l} \dfrac{1}{\sin(c\pi/q)^2} = \dfrac{2}{q^2}\sum_{c = \tau + 1}^{\floor{q/2}} \dfrac{1}{\sin(c\pi/q)^2}
    \]
    As the $\sin$ function is concave downward from $0$ to $\pi/2$, the function is lower-bounded by the line that connects the origin and the point at $(\pi/2,1)$, i.e., $\sin{t} \geq 2t/\pi$.
    \[
        \dfrac{2}{q^2}\sum_{c = \tau + 1}^{\floor{q/2}} \dfrac{1}{\sin(c\pi/q)^2} \leq \dfrac{2}{q^2}\sum_{c = \tau + 1}^{\floor{q/2}} \dfrac{q^2}{4c^2} = \dfrac{1}{2}\sum_{c = \tau + 1}^{\floor{q/2}} \dfrac{1}{c^2}
    \]
    Because $1/t^2$ is a decreasing function for $t > 0$, we can apply the integral bound to obtain
    \[
        \dfrac{1}{2} \sum_{c = \tau + 1}^{\floor{q/2}} \dfrac{1}{c^2} \leq \dfrac{1}{2}\int_{\tau}^{q/2} \dfrac{1}{c^2} dc \leq \dfrac{1}{2\tau}. \hfill\qedhere
    \]
\manualqed\end{proof}

\heavyfourier*
\begin{proof}
From \cref{prop:fourier_coeff}, using $\sin t \geq 2t/\pi$ in the relevant range,
    \begin{align*}
        \sum_{c \in \cW_h} \abs{\hat{f}(c)} \leq \sum_{c \in \cW_h} \dfrac{1}{q|\sin(c\pi/q)|} 
        &= 2 \sum_{c=1}^\tau \dfrac{1}{q \sin(c\pi/q)}\\
        &\leq \sum_{c = 1}^{\tau} \dfrac{1}{c}
 \leq 1 + \int_{1}^{\tau} \dfrac{1}{c} dc \leq 1 + \ln \tau. \tag*{\qedhere}
    \end{align*}
\manualqed\end{proof}

\subsection{Learning Significant Fourier Coefficients}
\label{sec:learn_fourier_coefficients}

With \cref{lemma:fourier_coeff}, we know that there exists a small $|c|$ such that the Fourier coefficient $\hat{\algP}(c \cdot x)$ is large, say at least $\alpha$. To find $c \cdot x$, we use the Learning Significant Fourier Coefficients algorithm from Akavia, Goldwasser and Safra~\cite{AGS03} to find all Fourier Coefficients with magnitude at least $\alpha$.

First, we show that the number of such coefficients is at most $1/\alpha^2$.
\begin{proposition}
\label{prop:list-size}
    Let $\alpha \in (0,1]$, $\cW = \{\chi \in \Integer{q}^n \ |\ \abs{\hat{\algP}(\chi)} \geq \alpha \}$, then $|\cW| \leq 1/\alpha^2$
\end{proposition}

\begin{proof}
    Since $\algP: \Integer{q}^{n} \to \{-1, 1\}$, $\expec{y \gets \Integer{q}^n}{\algP(y)^2} = 1$. So
    \[
     |\cW|\alpha^2 = \sum_{\chi \in \cW} \alpha^2   \leq \sum_{\chi \in \cW}\hat{\algP}(\chi)^2 \leq \sum_{\chi \in \Integer{q}^n}\hat{\algP}(\chi)^2 = \expec{y \gets \Integer{q}^n}{\algP(y)^2} \leq 1
    \]
\manualqed\end{proof}

\begin{theorem}[\cite{AGS03}]
    \label{thm:ags}
    Let $\alpha \in (0,1]$. There exists an algorithm that, given a function $\algP: \Integer{q}^{n} \to \{-1, 1\}$ and $\alpha$, outputs $\cW = \{\chi \in \Integer{q}^n \ |\ \abs{\hat{\algP}(\chi)} \geq \alpha \}$, except with probability $o(\alpha)$, and runs in time $\poly(n,\log q, 1/\alpha)$.
\end{theorem}

\subsection{Proof of \cref{theorem:main}}

From \cref{lemma:dist_adv} and \cref{cor:correlation}, we get an algorithm $\algP$ such that
\[
    \exp{(\algP(y, A, \cR(Ax)) - \mu_P)(S(y, x) - \mu_S)} \geq  \dfrac{B\eps}{2qt}
\]
Let $\eps^\prime = (B\eps/2qt)$. With probability at least $\eps^\prime/8$ we get a $\eps^\prime/2$-good instance $(A, x, \cR(Ax))$. In a good instance, \cref{lemma:fourier_coeff} says that there exists $c \in \Integer{q}$ such that $c \neq 0$, $|c| \leq \ceil{32/(\eps^\prime)^2}$ and 
\[
    \abs{\hat{\algP}(c \cdot x)} \geq \dfrac{\eps^\prime}{16 \ln (14/\eps^\prime)} = \alpha.
\]
Then, we run the \cite{AGS03} algorithm from \cref{thm:ags} on $\algP$ with $\alpha$ to obtain a list that contains $c \cdot x$. By \cref{prop:list-size}, this list has size at most $\tilde{O}(1/(\eps^\prime)^2)$. The final runtime of the algorithm is $\poly(n,m,\log q, 1/\eps^\prime, T)$, where $T$ is the running time of the algorithm for the Decision Noisy Linear Equations problem. 

The success probability of the algorithm depends on whether we obtain an $\eps^\prime/2$-good instance and whether the \cite{AGS03} algorithm from \cref{thm:ags} returns the list correctly. By the union bound, the algorithm will succeed with probability $\eps^\prime/8 - o(\eps^\prime) = \Omega(\eps^\prime)$
\section{Corollaries}
\label{sec:corollaries}
We can derive corollaries from \cref{theorem:main} for the specific noise functions corresponding to LWE and LWR. From \cref{prop:coprime_secret}, the total variation distance of $U(\cP(\Integer{q}^n))$ and $U(\Integer{q}^n)$ is small. Therefore, we write the following with the distribution of the secret being $U(\Integer{q}^n)$ instead.

\begin{corollary}[Partial Secret Recovery for LWE/LWR]
    \label{corollary:LWE-LWR}
    Consider any $n,m,q \in \mathbb{Z}^+$, and any $\eps \in (0, 1]$. Consider the LWE/LWR problem over $\Integer{q}$ with sample count $m$, secret dimension $n$.
    \begin{itemize}[itemsep=0pt]
        \item Under LWE noise with parameter $\sigma \in \Reals$, $\sigma < q$, set $\eps^\prime = \sigma\eps^2/q\sqrt{m}$, or,
        \item Under LWR noise with rounding modulus $p \in \Integer{}^+$, $p < q$, set $\eps^\prime = \eps^2/pm$
    \end{itemize}
    Suppose there exists an algorithm for the corresponding Decision LWE/LWR that runs in time $T$ and has advantage $\eps$.

    Then there exists an algorithm that, given the problem instance, outputs a list of size $\tilde{O}(1/(\eps^\prime)^2)$ that contains $c \cdot x$, where $x$ is the solution, for some $c \in \Integer{q} \setminus \{0\}$ such that $|c| \leq \tilde{O}(1/(\eps^\prime)^2)$. This algorithm runs in time $\poly(n, m, \log q, 1/\eps^\prime,T)$ and successfully outputs such a list with probability $\tilde{\Omega}(\eps^\prime)$.
    
\end{corollary}

\begin{proof}
    For LWE, by \cref{prop:LWE_noise_flooding}, set $B = \min(2\eps\sigma/\sqrt{m}, \floor{q/4 - 1})$. For LWR, by \cref{prop:LWR_noise_resilient}, set $B = \min(q\eps/pm, \floor{q/4 - 1})$. The rest follows from \cref{theorem:main}.  
\end{proof}

We note that, assuming that the search problem and decision problem are defined over the same secret distribution, our reduction also applies to secret distributions other than uniform over $\Integer{q}^m$, as long as the secret $x$ satisfies 
\[ \gcd(x_1, \dots, x_n, q) = 1 \]
with very high probability.

\subsection{Exact Secret Recovery}

So far, we have shown that our algorithm can compute a list of vectors that contains some small multiple of the secret. Under certain conditions, the exact secret can be recovered from this list. We describe a few sufficient conditions for the recovery.

\medskip
The first requirement for recovery is the existence of an efficient verification algorithm $\algV$ for the Noisy Linear Equations problem.

\begin{definition}[Verification Algorithm] A \emph{verification algorithm} $\algV$ for a Noisy Linear Equations problem is an algorithm in which the following holds with high probability over $A, x$, and $\mathcal{R}$:
\[
    \forall x^\prime \in \Integer{q}^n, \algV(x^\prime, A, \cR(Ax)) =
    \begin{cases}
        1 &\text{if}\quad x^\prime = x\\
        -1 &\text{if}\quad x^\prime \neq x
    \end{cases}
\]
\end{definition}
For typical Noisy Linear Equations problems, such a verification algorithm exists when $m$ is large enough, and the error is small enough for unique recovery. In particular, a verification algorithm for LWE exists when $n \log q  \ll m \log (q/\sigma)$, while it exists for LWR when $n \log q \ll m \log p$.

\medskip
The next requirement is to be able to recover $x \in \Integer{q}^n$ given $c \cdot x$ and $c$. 

\begin{definition}[Recovery Algorithm]
    A \emph{recovery algorithm} $\algG$ for a Noisy Linear Equations problem is an algorithm such that, for all $x$, and $c \in \Integer{q} \setminus \{0\}$:
    \[
        \algG(c \cdot x, c) = x
    \]
\end{definition}

Some settings were such a $\algG$ exists are when $q$ is a prime, or if we are guaranteed that $c \cdot \max_{i} x_i \leq q$. For example, when the secret $x$ is uniform over $\{0, 1\}^{n}$ and  $|c| < q$ (Binary secret LWE problem \cite{DBLP:conf/stoc/BrakerskiLPRS13}).

\medskip
Since the algorithm's output list size is $\tilde{O}(1/(\eps^\prime)^2)$ and $|c| \leq \tilde{O}(1/(\eps^\prime)^2)$, in time $\tilde{O}(1/(\eps^\prime)^4)$ we can try all candidates $c$ on every element in the output list. If $\algG$ returns $x$ on any one of them, the verification algorithm $\algV$ would tell us that we got the right answer. 

Suppose a verification algorithm does not exist, but $\algG$ exists. We can still perform recovery at the cost of lowering the success probability by randomly choosing an element from the output list and randomly choosing a $c$. The success probability of the search algorithm would then become $\Omega((\eps^\prime)^5)$.

\begin{corollary}[Search-to-Decision Reduction for LWE/LWR]
    \label{theorem:search-to-decision}
    Consider any $n,m,q,p \in \mathbb{Z}^+$, a noise function $\cR: \Integer{q}^{m} \to \Integer{p}^{m}$, and any $\eps \in (0, 1]$. For the LWE/LWR over $\Integer{q}$ with sample count $m$, secret dimension $n$,
    \begin{itemize}[itemsep=0pt]
        \item Under LWE noise with parameter $\sigma \in \Reals$, $\sigma < q$, set $\eps^\prime = \sigma\eps^2/q\sqrt{m}$
        \item Under LWR noise with rounding modulus $p \in \Integer{}^+$, $p < q$, set $\eps^\prime = \eps^2/pm$
    \end{itemize}
     Suppose there exists a \emph{Decision algorithm} that runs in time $T_\algD$ and has advantage $\eps$, and a \emph{Recovery algorithm} that runs in time $T_\algG$. 
     
     Then there exists a \emph{Search algorithm} that runs in time $\poly(n, m, \log q, 1/\eps^\prime, T_\algD, T_\algG)$ with success probability at least $\tilde{\Omega}((\eps^\prime)^5)$.

     Furthermore, if a \emph{Verification algorithm} that runs in time $T_\algV$ exists, there exists a \emph{Search algorithm} that runs in time $\poly(n, m, \log q, 1/\eps^\prime, T_\algD, T_\algV, T_\algG)$ with success probability at least $\tilde{\Omega}(\eps^\prime)$. Specifically, this applies to LWE when $n \log q  \ll m \log (q/\sigma)$, $\sigma \ll o(p)$ and LWR when $n \log q \ll m \log p$, $p \ll q$
\end{corollary}

\iflncs
\section{Acknowledgments}
AI tools were used as typing assistants for grammar and basic editing.
\else
\section{Acknowledgments}
This work was supported by the National Research Foundation, Singapore, under its NRF Fellowship programme, award no. NRF-NRFF14-2022-0010 and by an NSERC Discovery Grant. AI tools were used as typing assistants for grammar and basic editing.
\fi

\iflncs
\bibliographystyle{splncs04.bst}
\else
\bibliographystyle{alpha}
\fi
\bibliography{refs}

\appendix
\section{Supplementary Proofs}
\subsection{Hybrid Argument}

\hybridargument*

\begin{proof}[Proof of \cref{prop:hybrid_argument}.]
    \label{proof:hybrid_argument}
    Using $TV\Big(\cH_t, (A, \cR(b))\Big) \leq \eps^2/2qt$ from \cref{claim:terminal_hybrid},
    \begin{align*}
        &\abs{\prob{i \gets [0, t-1]}{D(\cH_{i}) = 1} - \prob{i \gets [0, t-1]}{D(\cH_{i+1}) = 1}}\\ 
        = &\abs{\sum_{i = 0}^{t-1} \dfrac{1}{t}\Pr[D(\cH_{i}) = 1] - \dfrac{1}{t}\Pr[D(\cH_{i+1}) = 1]}\\
        = &\abs{\dfrac{1}{t}(\Pr[D(\cH_{0}) = 1] - \Pr[D(\cH_{t}) = 1])} \\
        \geq &\eps/t - \eps^2/2qt \geq \eps/2t \hfill
    \end{align*}
\end{proof}
\end{document}